\documentclass[nonacm, sigconf]{acmart}

\usepackage{amsmath,amsfonts}
\usepackage{amsthm}
\usepackage{graphicx}
\usepackage{textcomp}
\usepackage{caption}    
\usepackage{multirow}
\usepackage{booktabs}
\usepackage{hyperref}
\usepackage{subcaption} 
\usepackage{algorithm}
\usepackage[noend]{algpseudocode}
\algrenewcommand{\algorithmicrequire}{\textbf{Input:}}
\algrenewcommand{\algorithmicensure}{\textbf{Output:}}
\newtheorem{theorem}{Theorem}
\AtBeginDocument{%
  }

\begin{document}

\title{Efficient MoE Inference with Fine-Grained Scheduling of Disaggregated Expert Parallelism}

\renewcommand{\shorttitle}{FinDEP}

\author{
Xinglin~Pan$^{1}$, Shaohuai~Shi$^{2}$, Wenxiang~Lin$^{2}$, Yuxin~Wang$^{3}$, Zhenheng~Tang$^{4}$, Wei~Wang$^{4}$, Xiaowen~Chu$^{1,4}$}
\affiliation{$^{1}$The Hong Kong University of Science and Technology (Guangzhou) \country{China}} 
\affiliation{$^{2}$Harbin Institute of Technology, Shenzhen \country{China}}
\affiliation{$^{3}$Hong Kong Baptist University \country{Hong Kong SAR}}
\affiliation{$^{4}$The Hong Kong University of Science and Technology \country{Hong Kong SAR}}

\sloppy
\begin{abstract}
 The mixture-of-experts~(MoE) architecture is commonly employed in contemporary large language models~(LLMs) due to its advantage of scaling model size with a sublinear increase in computational demand. Nevertheless, the inference of MoE models demands substantial memory, making it memory-intensive in attention layers due to the necessity of accessing key-value (KV) caches and in expert layers, utilizing only a limited number of experts. 
Recent studies attempt to utilize disaggregated expert parallelism (DEP) to distribute attention and experts to two dedicated GPU groups, the attention group (AG) and the expert group (EG), to improve inference efficiency. However, the existing DEP has limited support for modern MoE models with shared experts, and it under-explores task scheduling in both GPU groups, which have complex communication and computation tasks, leading to suboptimal inference performance.
To address these issues, we propose FinDEP, a \underline{fin}e-grained task scheduling algorithm for DEP with maximal task overlap to improve the inference throughput of MoE models.
FinDEP integrates our three proposed key innovations: 1) partitioning intensive computation and communication tasks to multiple smaller tasks in both AG and EG to enable fine-grained task pipelining w/ or w/o shared experts, 2) formulating an optimization problem to the fine-grained task scheduling that should support different task partition granularity and ordering, and 3) developing an efficient solution to the optimization problem which contains a huge solution space to derive the near-optimal task schedule of DEP. 
Experiments are conducted on four types of GPU systems with two representative MoE backbones, DeepSeek-V2 and Qwen3-MoE. Experimental results show that FinDEP improves inference throughput by up to 1.61$\times$ over state-of-the-art methods. Notably, on a 32-GPU system, FinDEP still achieves a significant speedup of up to $\mathbf{1.24\times}$, demonstrating its efficiency at large scales.
\end{abstract}



\maketitle

\section{Introduction}

Large language models~(LLMs) are scaling rapidly, and so are their computational costs. For instance, models like Falcon~\cite{falcon} with 180 billion parameters and Llama-3.1~\cite{llama3} with 405 billion parameters exemplify this trend. Mixture-of-Experts (MoE) architectures~\cite{mixtral,MOE_LSTM, GShard} address this challenge by activating only a subset of the model's expert components for each input.
This makes it possible to build much larger models without making training or inference more expensive.
Recent MoE-based LLMs, such as DeepSeek-V3~\cite{deepseek_v3} and Qwen3-MoE~\cite{qwen3}, show that this design can create highly capable models that are still fast and cheap to use. 
As a result, MoE has become a key technique for building future LLMs in a way that balances power and efficiency.

Despite the advantages, running inference on large MoE models remains challenging~\cite{deepseek_v3,cloud384, moe_lightning,adapmoe,fiddler} due to its extensive memory requirement to hold all experts in the MoE layers and key-value (KV) caches in the attention layers. As a result, distributing the MoE model across multiple GPUs has been a common practice~\cite{GShard,deepseek_v3,MegaScale_Infer} for efficient inference through expert parallelism (EP), which assigns experts across GPUs.

\begin{figure}[!t]
    \centering
    \includegraphics[width=\linewidth]{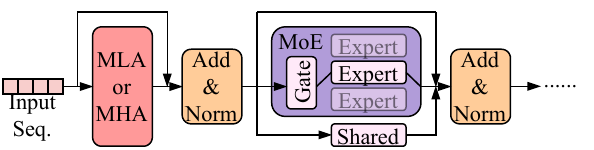}
    \caption{A typical structure of an MoE model. MLA refers to Multi-Head Latent Attention~\cite{deepseek_v2}, while MHA denotes Multi-Head Attention~\cite{mha}. The ``Shared'' block indicates one shared expert or several shared experts, which may be optional depending on the MoE configuration.}
    \label{fig:moe}
\end{figure}

Recent research~\cite{MegaScale_Infer, stepfun, huawei_ae} suggests distributing attention layers and expert layers onto distinct GPUs through disaggregated expert parallelism (DEP)\footnote{Also referred to as Attention-FFN Disaggregation (AFD); we adopt the term DEP following \cite{MegaScale_Infer}.}, due to the different computational and memory access patterns of attention layers and expert layers. This approach enables the modules to scale independently while optimizing the use of various hardware capabilities. In DEP, a multi-GPU system is divided into two groups: the attention group (AG), responsible for storing all attention layers, and the expert group (EG), which holds all non-shared experts. It is important to mention that in certain MoE models like DeepSeek-V3~\cite{deepseek_v3}, shared experts within the MoE layer are often placed in the AG as they need to be processed by all input tokens. 
The dependency between attention and expert layers is substantial, as each attention layer's output serves as the input for the subsequent expert layer, which then outputs to another attention layer as shown in Fig.~\ref{fig:moe}. Consequently, DEP necessitates bidirectional communication: from AG to EG (A2E) and the reverse (E2A). \textit{Data dependencies and communication overhead easily lead to the GPU computational resources idle, thereby limiting inference efficiency.}

Existing optimizations try to alleviate the GPU idle duration of DEP via 1) overlapping computation and communication tasks to reduce the communication time with the ping-pong pipeline (PPPipe) algorithm proposed in MegaScale-Infer~\cite{MegaScale_Infer} or 2) offloading communication tasks to CPU resources to enable overlaps between CPU communications and GPU computations in StepMesh~\cite{stepfun}. These techniques enable only coarse-level task scheduling by dividing a mini-batch into several micro-batches. \textit{As a result, different tasks from these micro-batches can be executed in a pipeline fashion, but this does not sufficiently hide A2E/E2A communications, leading to suboptimal inference efficiency.} Moreover, certain cutting-edge MoE models such as DeepSeek series~\cite{dai2024deepseekmoe,deepseek_v2,deepseek_v3} introduce shared experts within the MoE layer, which are required to compute for every input token, similar to the attention layer, leading to increased GPU idle time.

In this paper, we propose FinDEP, a fine-grained task scheduling framework for MoE inference with DEP to address the above two efficiency problems by three key innovations. 
(1) We partition time-consuming tasks including computations in EG, communications in A2E and E2A, and computations in AG into smaller tasks by splitting each task's input tensor into several segments (denoted as $r$). This partitioning of the tensor creates $r$ smaller tasks per original task, allowing for dynamic scheduling aimed at improving the throughput for MoE models, regardless of whether they have shared experts.
(2) Intuitively, increasing $r$ allows greater parallelization for enhanced overlapping. However, this also increases the launch overheads associated with executing tasks, such as kernel dispatch on GPUs and communication startup costs. Thus, a balance must be built between the advantages of overlapping and the execution overheads. Consequently, we construct performance models for computation tasks in AG and EG and their A2E/E2A communication tasks. Using these models, we establish an optimization problem to characterize the DEP inference time with fine-grained task scheduling, including task ordering and tensor partition granularity.
(3) We develop an efficient algorithm to find the near-optimal solution to the formulated optimization problem with a polynomial time complexity, thus avoiding the very time-consuming brute-force search on the huge solution space. 

We conduct extensive experiments on four GPU systems with two representative MoE model backbones, DeepSeek-V2 (with shared experts) and Qwen3-MoE (without shared experts). Experimental results show that our FinDEP achieves speedups of upto 1.61$\times$ over the best-configured PPPipe algorithm in MegaScale-Infer. Furthermore, on the 32-GPU system, FinDEP consistently provides a speedup of up to $\mathbf{1.24\times}$. Beyond peak throughput, we also confirm the computational efficiency of our fine-grained task scheduling solver. Our solver is highly efficient, taking less than one second to compute the near-optimal configuration. This minimal overhead enables real-time adaptation to dynamic workloads, which is crucial for maximizing throughput in online serving environments with dynamically varying sequence lengths and batch sizes.

\section{Background and Motivations}\label{sec:background}
This section provides an overview of background concepts, followed by a summary of the motivations for this research. For clarity, Table~\ref{tab:notations} offers a summary of the frequently used notations throughout the paper.

\begin{table}[!t]
\centering
\caption{Notations.}
\resizebox{\linewidth}{!}{%
\label{tab:notations}
\begin{tabular}{ll}
\toprule
\textbf{Name} & \textbf{Description}  \\
\midrule
$P$ & \# of GPUs in the cluster. \\
$ag$ & Size of attention group~(AG). \\
$eg$ & Size of expert group~(EG). \\
$m_a$ & \# of samples per micro-batch per GPU in AG. \\
$m_e$ & \# of tokens per micro-batch per expert. \\
$S$ & Sequence length of each sample. \\
$E$ & Total number of global experts. \\
$T$ & Total number of layers. \\
$M$ & Embedding size for each token. \\
$H$ & Hidden size of the feed-forward layer within experts. \\
$top_k$ & \# of experts activated per token. \\
$r_1$ & Pipeline degree of the AG. \\
$r_2$ & Fine-grained pipeline degree of the EG. \\
\bottomrule
\end{tabular}%
}
\end{table}

\subsection{MoE Layer}


MoE models replace each dense feed-forward network (FFN) in transformers with sparsely activated FFNs (or experts) by the MoE layer, as shown in Fig.~\ref{fig:moe}. 
Each token is routed to $k$ experts via a gating function: the gate computes routing scores over all experts, applies a softmax function, and selects the top-$k$ experts for each token~\cite{Auxiliary-loss-free}. The input is partitioned accordingly, with each expert processing only its assigned tokens. Some implementations include a shared expert that processes all tokens~\cite{deepseek_v2, deepseek_v3, Rajbhandari2022deepspeed_moe}, while it is optional in some implementations like Qwen3-MoE~\cite{qwen3}. The layer output of the MoE layer is the aggregated contributions from selected experts (and the shared expert if present). 

\subsection{Disaggregated Expert Parallelism and Ping-pong Pipeline}\label{subsec:dep}

\begin{figure}[t]
    \centering
    \includegraphics[width=0.9\linewidth]{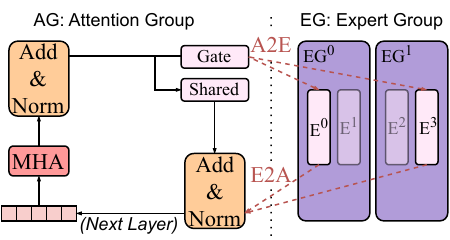}
    \caption{An illustration of DEP. GPUs are partitioned into two groups: AG and EG. AG handles the attention and shared expert computation, while EG handles experts computation. }
    \label{fig:dep}
\end{figure}





\textbf{Disaggregated Expert Parallelism.}  
Disaggregated Expert Parallelism (DEP) is a novel parallelization strategy specifically tailored for the high-throughput, low-latency inference of large MoE-based models. Its foundational principle is the physical separation and independent allocation of core model components across distinct GPU groups. This partitioning divides the available hardware into two dedicated functional units: the Attention Group (AG) and the Expert Group (EG) as shown in Fig.~\ref{fig:dep}. The AG is dedicated to storing and processing the standard components of the Transformer block, including the Self-Attention layers and the Shared Expert (if present), (i.e., components that are densely activated across all tokens). Conversely, the EG houses the entire set of sparse MoE experts, distributed across its constituent devices. This structural disaggregation enables to independently scale the computational resources for each module based on specific memory and computational bottlenecks, a key advantage over monolithic parallel approaches.

A key architectural benefit of DEP is the elimination of intra-group communication overhead. Within the AG, parameters are fully replicated, allowing each device to operate independently without costly collective operations (e.g., All-Reduce). Similarly, within the EG, the inherent sparsity of the token-to-expert routing ensures that an activated expert's computation is confined to a single GPU. This confinement prevents the necessity of communication between expert devices.

The necessary collective communication occurs solely between the two groups through two defined communication phases: 1) Attention-to-Expert (A2E), where tokens processed by the AG are routed to the appropriate expert(s) in the EG, and 2) Expert-to-Attention (E2A), where the expert outputs are gathered and returned to the AG for subsequent layers. 
This disaggregation enables independent scaling of computational resources for each module and the development of tailored parallel strategies~\cite{MegaScale_Infer,huawei_ae}.
The sequential execution of MoE inference with DEP is illustrated in Fig.~\ref{fig:motivation2}(a). Due to the data dependency between modules, the EG remains idle until the AG completes its forward pass and dispatches tokens via A2E communication. Conversely, the AG must wait idly for the EG to finish processing and return results via E2A before it can proceed to the next layer. While this disaggregation offers significant flexibility, this sequential handoff leads to significant device idle time in a naive implementation, as computational resources in one group are consistently underutilized while waiting for the other group to fulfill its part of the pipeline.


\textbf{Ping-Pong Pipeline Parallelism.} To rigorously address the device idle time inherent in the sequential dependency between the Attention Group and the Expert Group, the Ping-Pong Pipeline Parallelism~(PPPipe) algorithm~\cite{MegaScale_Infer} serves as a specialized micro-batch scheduling strategy that enables the concurrent utilization of both AG and EG as shown in Fig.~\ref{fig:motivation2_0}. Specifically, PPPipe divides the input mini-batch into $r_1$ micro-batches (e.g., $r_1=2$ in Fig.~\ref{fig:motivation2_0}) to allow GPUs in EG to begin computations without waiting for the full output of the mini-batch. Thus, in PPPipe, AG and EG computation tasks can be executed in parallel, and the communication tasks can also be overlapped with the computation tasks, thus improving the inference throughput. 

\subsection{Motivations}\label{subsec:motivation}

\begin{figure}[t]
    \centering
    \begin{subfigure}{0.9\linewidth}
        \centering
        \includegraphics[width=\linewidth]{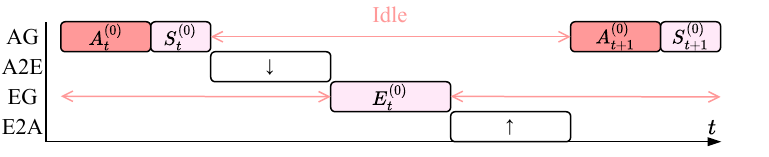}
        \caption{Naive DEP without pipelining}
        \label{fig:motivation2_-1}
    \end{subfigure} \\
    \vspace{2pt}
    \begin{subfigure}{0.9\linewidth}
        \centering
        \includegraphics[width=\linewidth]{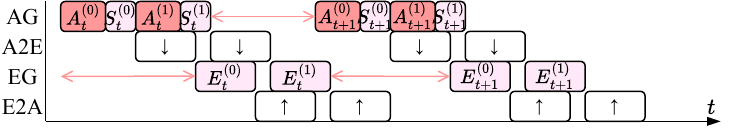}
        \caption{PPPipe with $r_1=2$ (shared expert is a part of attention).}
        \label{fig:motivation2_0}
    \end{subfigure} \\
    \vspace{2pt}
    \begin{subfigure}{0.9\linewidth}
        \centering
        \includegraphics[width=\linewidth]{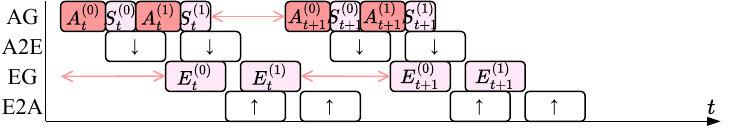}
        \caption{FinDEP with $r_1=2$ and $r_2 = 1$.}
        \label{fig:motivation2_1}
    \end{subfigure} \\
    \vspace{2pt}
    \begin{subfigure}{0.9\linewidth}
        \centering
        \includegraphics[width=\linewidth]{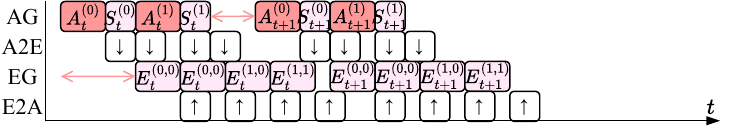}
        \caption{FinDEP with $r_1=2$ and $r_2 = 2$.}
        \label{fig:motivation2_2}
    \end{subfigure}
    \caption{Timeline of naive DEP, PPPipe, and our FinDEP.}
    \label{fig:motivation2}
    \vspace{-2pt}
\end{figure}
While PPPipe in MegaScale-Infer~\cite{MegaScale_Infer} allows AG and EG tasks to be pipelined to reduce the GPU idle time, it is still suboptimal due to the following limitations.

\textit{Computation tasks of the shared expert are not well scheduled.} In PPPipe~\cite{MegaScale_Infer}, it assumes that there is no shared expert in AG, which does not support recent MoE models like DeepSeek-V3. Built atop PPPipe, one can support including the shared expert by regarding it as a part of attention, since both attention and the shared expert should process all input tokens. As shown in Fig.~\ref{fig:motivation2_0}, A2E can only begin after the completion of the shared expert computation. However, the computation of experts in AE has no data dependency with the shared expert; thus, they can also overlap. This means the computation tasks of the shared expert should be well-scheduled to achieve better efficiency.

\textit{Micro-batch level pipelining is insufficient to overlap tasks fully.} 
Existing DEP implementations, including PPPipe, overlook potential performance gains from overlapping communication and computation between attention and expert modules.
While existing solutions focus on overlapping attention and expert computations, they underestimate additional benefits from overlapping A2E or E2A communication with expert computation.
This overlap allows the GPU-resident expert module to begin computation earlier, potentially improving utilization and throughput~\cite{shi2023pipemoe}.
We show an example in Fig.~\ref{fig:motivation2_2}, where dividing the expert into two micro-batches reduces end-to-end execution time. 
However, introducing pipelining can also incur kernel launch overhead, which in some cases may increase rather than reduce expert computation time, thereby worsening bottlenecks.
To address this trade-off, a modern and adaptive pipelining degree is required to balance early computation benefits and kernel launch costs. Balancing overlap benefits and launch costs is crucial for further exploiting throughput potential.

\textit{Huge search space to find an optimal schedule.} 
The integration of shared experts and fine-grained pipeline settings significantly expands the search space for optimal configurations. Consequently, this expansion stems from numerous interdependent design choices, including pipeline degrees, microbatch sizes, and task orders. Similarly, decisions regarding the fine-grained degree and micro-size of experts, alongside the configurations of AG and EG, further compound this complexity. Due to such entanglement, brute-force enumeration becomes impractical. Therefore, an adaptive and efficient algorithm is necessary to explore the design space to find the optimal solution efficiently.

In this paper, we aim to address the above three issues by proposing FinDEP to partition tensors for fine-grained task scheduling with the support of shared experts. Thus, we split the attention input for each GPU along the batch dimension to enable a micro-batch level pipeline. The number of pipelines is denoted by $r_1$, and the micro-batch size per GPU is denoted by $m_a$. Since there are no data dependencies, the shared expert and A2E of each micro-batch can run in parallel, as shown in Fig.~\ref{fig:motivation2_1}. Other task orders are discussed in the next section. Unlike the attention part, which involves interactions between tokens (i.e., intra-sequence), the expert part processes samples token by token. Based on this, we can further partition along the token dimension. The pipeline degree is denoted by $r_2$, and $m_e$ represents the token processed by each expert. An example is shown in Fig.~\ref{fig:motivation2_2}. A primary challenge within FinDEP is to define the optimal problem (\S\ref{sec:problem}) and derive the optimal solution (\S\ref{sec:solution}), which we will present in the next two sections.


\section{Problem Formulation}\label{sec:problem}

The inference time of an MoE model under disaggregated expert parallel decomposes into three primary components: the computation time for the expert feed-forward networks, the computation time for the attention layers, and the communication overhead for transferring activations between the attention and expert groups. We formulate each component as a function of its workload and the underlying hardware characteristics.

\subsection{Execution Time Formulation}


First, for GEMM, we denote the time function as $t_{gm}(x, F)$, where $x = m \times k \times n$ represents the total FLOPs required for multiplying two matrices $A \in \mathbb{R}^{m \times k}$ and $B \in \mathbb{R}^{k \times n}$ on a GPU with peak floating-point performance $F$.

The second part of the computation involves the self-attention mechanism. The time function for the attention computation is denoted as \( t_{attn}(y, F) \), where \( y \) represents the total workload for self-attention, and \( F \) is the GPU's peak performance again. In this case, the workload is defined based on the dimensions of the query (\( Q \)), key (\( K \)), and value (\( V \)) matrices. These matrices have the following shapes: \( Q, K \in \mathbb{R}^{N_h \times B \times S \times D_k} \), and \( V \in \mathbb{R}^{N_h \times B \times S \times D_v} \), where \( N_h \) is the number of attention heads, \( B \) is the batch size, \( S \) is the sequence length, \( D_k \) is the dimensionality of the key, and \( D_v \) is the dimensionality of the value. The core computational burden comes from two GEMMs: the computation of the attention scores via \( QK^\top \), which has a complexity of \( N_h B S^2 D_k \), and the computation of the attention-weighted values \( \text{Attention}(QK^\top)V \), which has a complexity of \( N_h B S^2 D_v \). Therefore, the total workload for self-attention is \( y = N_h B S^2 (D_k + D_v) \).

Third, $t_{{c}}(z, eg, ag)$ denotes the communication time required for data transfer between GPUs. More specifically, it measures the time taken for $ag$ GPUs to send messages to $eg$ other GPUs. Here, $z$ represents the communication workload per machine, while $eg$ and $ag$ refer to the expert and attention group sizes, respectively.

For any given hardware and group configuration, the GPU performance \( F \) and group sizes (\( eg \), \( ag \)) are constant. This stability allows us to simplify the time functions, specifically \( t_{gm}(x) \), \( t_{attn}(y) \), and \( t_{a2e}(z) \), for convenience. These simplified functions provide a foundation for modeling the end-to-end performance of MoE systems.

\textbf{The Attention Part.} 
The attention component consists of a sequence of computational operations that include both GEMMs and attention. 
To illustrate this process, we consider the standard Multi-Head Attention~(MHA) layer as an example. In the $t$-th transformer layer, the input comprises hidden states $\mathbf{h}_t \in \mathbb{R}^{m_a \times S \times M}$. The total forward pass time for this layer, denoted $t_a(m_a)$, is a function of these dimensions and can be decomposed into the cumulative runtime of several GEMM and self-attention operations:
\begin{equation}
\begin{aligned}
t_a(m_a) =\ & 2 t_{{gm}}(m_a S M n_h d_k)  + 2 t_{{gm}}(m_a S M n_h d_v) \\
& + t_{{attn}}(m_a S^2 n_h (d_k + d_v)). \\
\end{aligned}
\label{eq:attn_detail}
\end{equation}
The coefficients 2 and 2 in the $t_{{gm}}$ terms account for the four linear projections required by the MHA operation: $Q$ and $K$ projections, and $V$ and Output ($O$) projections. Notably, other attention variants like MLA~\cite{deepseek_v2} can also be modeled using similar formulations involving $t_{{attn}}$ and $t_{{gm}}$, enabling unified analysis across various attention designs.

\textbf{The Shared Expert Part.}
The Shared Expert computation follows a structure similar to the attention layer, consisting of three primary linear projections: the gating projection, the up-projection, and the down-projection. For each expert $i$ ($1 \leq i \leq N_{{shared}}$), the gating and up-projections are represented by $W^{{gate}}_i$ and $W^U_i$, respectively, with dimensions $W^{{gate}}_i, W^U_i \in \mathbb{R}^{H \times M}$, while the down-projection is given by $W^D_i \in \mathbb{R}^{M \times H}$.

Each device in AG performs the shared expert transformations locally. The gating operation computes $\mathbf{z}^{{gate}}_{t,i} = W^{{gate}}_i \mathbf{h}_t$, the up-projection computes $\mathbf{z}^{u}_{t,i} = W^{U}_i \mathbf{h}_t$, and the down-projection computes $\mathbf{z}^{d}_{t,i} = W^{D}_i \text{Swish}(\mathbf{z}^{{gate}}_{t,i} \otimes \mathbf{z}^{u}_{t,i})$~\cite{swish},
where 
    $\text{Swish}(x) = \frac{x}{1+e^{-x}}$.
These operations result in outputs with dimensions $\mathbf{z}^{{gate}}_{t,i}, \mathbf{z}^{u}_{t,i} \in \mathbb{R}^{m_a \times S \times H}$ and $\mathbf{z}^{d}_{t,i} \in \mathbb{R}^{m_a \times S \times M}$, each taking $t_{{gm}}(m_a S M H)$ time.
The total computation time for the Shared Expert across $N_{{shared}}$ expert layers is the sum of all layers:
\begin{equation}
t_s(m_a) = 3N_{shared} t_{gm}(m_aSMH). \label{eq:shared_detail}    
\end{equation}


\textbf{The MoE Part.} 
The MoE layer employs conditional computation through a set of feed-forward networks, known as experts. The total number of $E$ experts is distributed across $eg$ devices. Each device is responsible for computing $E/eg$ distinct experts.
Each device in the expert group receives tokens, represented as $\mathbf{h}'_t \in \mathbb{R}^{(E/eg) \times m_e \times M}$, which are then partitioned along the first dimension into $E/eg$ slices: $\mathbf{h}'_{t,1}, \mathbf{h}'_{t,2}, \dots, \mathbf{h}'_{t,E/eg}$. Each slice, $\mathbf{h}'_{t,i} \in \mathbb{R}^{m_e \times M}$, is assigned to the corresponding local expert $i$. Here, $m_e$ denotes the number of tokens processed by a single expert. For each expert $i$, the computation involves a feed-forward network with weights including the up-projection $W^{U}_i \in \mathbb{R}^{H \times M}$, the gating projection $W^{{gate}}_i \in \mathbb{R}^{H \times M}$, and the down-projection $W^{D}_i \in \mathbb{R}^{M \times H}$, all of which reside on the assigned device.
The total computation time for each device is given by:
\begin{equation}
t_e(m_e) = 3(E / eg) t_{gm}(m_eMH).\label{eq:moe_detail}
\end{equation}


\textbf{A2E and E2A communication.} 
DEP employs two distinct communication operations: Attention-to-Expert~(A2E) and Expert-to-Attention~(E2A). We denote their respective communication times as $t_{a2e}$ and $t_{e2a}$. 
Since the communication workload is $z = E / eg \times m_e \times M$. We have:
\begin{equation}
t_{a2e}(m_e) = t_{c}(m_eEM/eg).\label{eq:a2e_detail}
\end{equation}
Due to the symmetric nature of communication in dual-workload topologies like PCIe or NVLink~\cite{pcie}, where data transfer occurs in different directions simultaneously, the time taken for A2E equals that for E2A, i.e., $t_{a2e}(m_e) = t_{e2a}(m_e)$.

\subsection{Optimization Problem Formulation}

For a given layer $t$, we define the timestamps that capture the start times of major computational and communication stages. Let $\tau_a^{(t, i)}$ represent the start time of the $i$-th attention segment within the $r_1$ pipeline, and $\tau_s^{(t, i)}$ the start time of the corresponding Shared Expert computation. The Expert computation within each pipeline segment is also divided into $r_2$ parts, and we denote the start time of the expert processing for the $i$-th $r_1$ slice and $j$-th $r_2$ token group as $\tau_e^{(t, i, j)}$. The communication timestamps are defined as $\tau_{a2e}^{(i, j)}$ and $\tau_{e2a}^{(i, j)}$, indicating the start times of the A2E and E2A communication phases, respectively.

Based on the above execution time formulations, we derive a set of timing constraints between key scheduling timestamps: $\tau_a^{(t, i)}$, $\tau_s^{(t, i)}$, $\tau_e^{(t, i, j)}$, $\tau_{a2e}^{(t, i, j)}$, and $\tau_{e2a}^{(t, i, j)}$. These constraints describe how different stages of computation and communication must be ordered to avoid conflicts and ensure data dependencies are satisfied. All constraints are represented as
\begin{equation}
\left\{
\begin{aligned} 
\tau_s^{(t', i')}, \tau_a^{(t', i')} &\notin [ \tau_a^{(t, i)}, \tau_a^{(t, i)} + t_a(m_a) ) \\
\tau_s^{(t', i')}, \tau_a^{(t', i')} &\notin [ \tau_s^{(t, i)}, \tau_s^{(t, i)} + t_s(m_a) ) \\
\tau^{(t',i', j')}_{a2e}, &\notin [ \tau^{(t,i,j)}_{a2e}, \tau^{(t,i,j)}_{a2e} + t_{a2e}(m_e)) \\
\tau^{(t',i', j')}_{e2a}, &\notin [ \tau^{(t,i,j)}_{e2a}, \tau^{(t,i,j)}_{e2a} + t_{e2a}(m_e)) \\
\tau^{(t',i', j')}_{e} &\notin [ \tau^{(t,i,j)}_{e}, \tau^{(t,i,j)}_{e} + t_{e}(m_e)) \\
\tau^{(t,i)}_s, \tau^{(t,i,j)}_{a2e} &\geq \tau^{(t,i)}_a + t_a(m_a) \\
\tau^{(t,i,j)}_e &\geq \tau^{(t,i,j)}_{a2e} + t_{a2e}(m_a) \\
\tau^{(t,i,j)}_{e2a} &\geq \tau^{(t,i,j)}_{e} + t_e(m_e) \\
\tau^{(t + 1,i)}_a &\geq \max (\tau^{(t,i,j)}_{e2a} + t_{e2a}(m_e), \tau^{(t,i)}_s + t_s(m_a)) \\
 m_e \cdot r_2 \cdot E &= m_a \cdot ag \cdot top_k \cdot S
\end{aligned}.
\right. \label{eq:limitation}
\end{equation}
The first five rules prevent different stages from using the same hardware at the same time. This avoids resource conflicts. Rules 6 to 9 ensure that each stage starts only after the previous one finishes in the same micro-batch. The final rule ensures that all data is processed accurately without any loss.

Our goal is to maximize the throughput of the disaggregated MoE pipeline by jointly optimizing the pipeline degrees and token partition sizes. This leads to the following optimization formulation:
\begin{equation}
\begin{aligned}
\max_{\substack{r_1, m_a \\ r_2, m_e}} 
& \frac{r_1 \cdot m_a \cdot ag}{\max ( \tau^{(T, r_1)}_s + t_s(m_a),\; \tau^{(T, r_1, r_2)}_{e2a} + t_{e2a}(m_e) )}  \\
\text{s. t.}  
& \quad \text{constraints in Eq.~\eqref{eq:limitation}}.  
\end{aligned}\label{eq:final-objective}
\end{equation}

For any fixed choice of \(r_1, m_a, r_2, m_e\), the remaining task is to assign start times \(\tau\) that satisfy the constraints in Eq.~\eqref{eq:limitation} and minimize the makespan \(\max\bigl(\tau_s^{(T, r_1)} + t_s(m_a),\;\tau_{e2a}^{(T, r_1, r_2)} + t_{e2a}(m_e)\bigr)\). This scheduling subproblem is a variant of the job‑shop problem: each operation (attention, shared, A2E, expert, E2A) runs on a dedicated machine (resources) and the operations of each micro‑batch follow the precedence graph implied by rules 6–9. It is well known that job‑shop scheduling is NP‑hard even with three machines. Our model involves four distinct resources (e.g., AG, EG, A2E, and E2A), therefore the subproblem is NP‑hard. Consequently, the overall problem, which additionally optimizes over the integer parameters \(r_1, m_a, r_2, m_e\), is also NP‑hard, because a polynomial‑time algorithm for the overall problem would yield a polynomial‑time solution for the NP‑hard subproblem by fixing those parameters appropriately.

\section{Solution}\label{sec:solution}

To solve the above problem, we need to explicitly determine the communication and computation times. Thus, we need to model the performance for a given communication or computation operation, so that we can predict their execution time with different sizes of input. In this section, we first build simple yet effective performance models for attention, GEMM computation, and A2E/E2A communication, then we derive the near-optimal solution to the problem of minimizing Eq.~\ref{eq:final-objective}.  

\subsection{Performance Models}









\textbf{Performance model of computation. }
Following \cite{shi2023pipemoe, gpu_launch, fsmoe}, we use a linear model (with bias) to represent computation. The model includes an intercept term ($\alpha_{gm}$) to account for fixed overheads, such as kernel launches and memory management, and a scaling factor ($\beta_{gm}$) to capture the increase in computational cost as the input size grows. The model is expressed as:
\begin{align}
    t_{gm}(x) &= \alpha_{gm} + \beta_{gm} x. \label{eq:gm} \\
    t_{attn}(y) &= \alpha_{attn} + \beta_{attn} y. \label{eq:attn} 
\end{align}

\textbf{Performance model of communication. }
For both A2E and E2A operations, the communication time can also be accurately described using a single $\alpha$-$\beta$ model. These operations are essentially reverse processes that share identical communication structures, allowing for a unified linear model. We define the communication time as:
\begin{equation}
t_{c}(z) = \alpha_{c} + \beta_{c} z \label{eq:a_r},
\end{equation}
where $z$ represents the input data size (bytes of elements communicated), $\alpha_{c}$ is the network startup time (overhead), and $\beta_{c}$ is the transmission time per byte, which is influenced by factors such as network bandwidth.

\textbf{Performance models of different layers. }
By substituting Eq.~\ref{eq:gm} and Eq.~\ref{eq:attn} into Eq.~\ref{eq:attn_detail}, we derive a simplified linear model for the performance of the MHA layer, expressed as: $t_a(m_a) = \alpha_a + \beta_a m_a$, where the new coefficients are defined as follows:
\begin{equation}
\alpha_a := 4 \alpha_{gm} + \alpha_{attn}
\end{equation}
and
\begin{equation}
\begin{aligned}
\beta_a := & \, \beta_{gm} ( 2 S M n_h d_k + 2S M n_h d_v) \\
           & + \beta_{attn} S^2 n_h (d_k + d_v) .
\end{aligned}
\end{equation}

For clarity and analytical tractability, we absorb all terms that do not vary with $m_a$ into the constants $\alpha_a$ and $\beta_a$. This yields a linear performance model that captures the contribution of $m_a$ to computation time.

Similarly, by substituting Eq.~\ref{eq:gm} into Eq.~\ref{eq:shared_detail}, we derive a simplified linear model for the performance of the shared expert layer, expressed as $t_s(m_a) = \alpha_s + \beta_s m_a$, where the new coefficients are defined as follows: $\alpha_s :=  3N_{shared} \alpha_{gm}$ and $\beta_s := 3N_{shared} \beta_{gm} SMH $.

Building on this methodology, the MoE layer's performance is modeled. Substituting Eq.~\ref{eq:gm} into Eq.~\ref{eq:moe_detail} and Eq.~\ref{eq:a_r} into Eq.~\ref{eq:a2e_detail}, and absorbing all terms not varying with $m_e$ into constants, we express its performance as $t_e(m_e) = \alpha_e + \beta_e m_e$. Here, $\alpha_e := (E / eg) \alpha_{gm}$ and $\beta_e := (E / eg) \beta_{gm} (MH)$. For A2E and E2A, the model is $t_{a2e}(m_e) = \alpha_{a2e} + \beta_{a2e} m_e$, where $\alpha_{a2e} := \alpha_c$ and $\beta_{a2e} := \beta_c EM/eg$. These linear models provide a tractable framework for analyzing and predicting the computational overhead of each architectural component, laying the groundwork for subsequent performance optimization. Although streamlined, this model effectively captures the dominant performance determinants of startup latency and workload-dependent scaling, and its fidelity is empirically validated in (\S\ref{sec:5.2}).

\subsection{Determine Task Order, $m_a$, $r_1$, $m_e$, and $r_2$.}


\begin{figure}[t]
    \centering

    \begin{subfigure}{0.9\linewidth}
        \centering 
        \includegraphics[width=\linewidth]{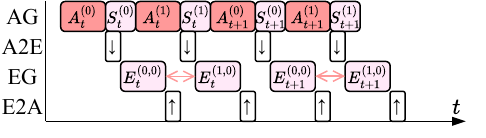}
        \caption{An example illustrating the limitations of \textit{ASAS}.}
        \label{fig:aass_e_as}
    \end{subfigure} \\
    
    \begin{subfigure}{0.9\linewidth}
        \centering
        \includegraphics[width=\linewidth]{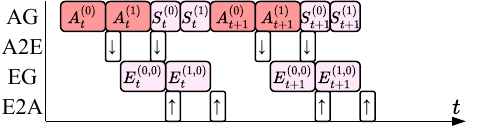}
        \caption{An example illustrating the advantages of \textit{AASS}.}
        \label{fig:asas_e_as}
    \end{subfigure} \\
    
    \begin{subfigure}{0.9\linewidth}
        \centering
        \includegraphics[width=\linewidth]{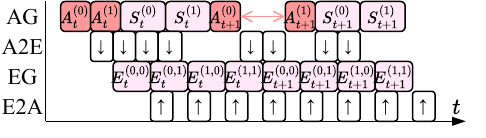}
        \caption{An example illustrating the limitations of \textit{AASS}.}
        \label{fig:aass_as_e}
    \end{subfigure} \\
    
    \begin{subfigure}{0.9\linewidth}
        \centering
        \includegraphics[width=\linewidth]{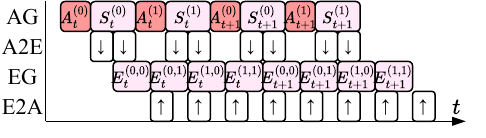}
        \caption{An example illustrating the advantages of \textit{ASAS}.}
        \label{fig:asas_as_e}
    \end{subfigure} 
    
    \caption{Comparative examples highlighting the advantages and limitations of \textit{AASS} and \textit{ASAS} scheduling strategies.}
    \label{fig:as}
\end{figure}

\textbf{Determine the order of Attention and Shared Expert.} We investigate the optimal execution order of attention and Shared Expert operations in AG by evaluating two primary scheduling strategies. The number of possible non-illness computing orders in a layer is given by $C(r_1 + r_1 - 1, r_1) = \frac{(2r_1 - 1)!}{(r_1!)((r_1 - 1)!)}$, which is cumbersome to verify one by one. However, we can observe that the advantages of more efficient computing are: (a) it allows for the earliest possible start of A2E communication, which helps utilize EG without idle time, and (b) it enables the use of AG (Attention Gate) without idle time.

We focus on the most representative strategies and explain why they are effective. The first, \textit{AASS} (Attention-All, Shared-All), processes all attention segments within the same layer before proceeding to all Shared Expert segments. The second, \textit{ASAS} (Attention-Shared-Alternating-Sequential), alternates between attention and Shared Expert operations.

As illustrated in Fig.~\ref{fig:as}, each schedule presents distinct advantages. The \textit{AASS} approach enables earlier initiation of A2E communication and expert computation, as evident when comparing Fig.~\ref{fig:aass_e_as} and Fig.~\ref{fig:asas_e_as}. Conversely, \textit{ASAS} improves GPU utilization by interleaving Shared Expert segments during periods in which attention-ready signals are pending, as shown in Fig.~\ref{fig:aass_as_e} and Fig.~\ref{fig:asas_as_e}.

To determine the better strategy, we independently identify the best-performing configuration for both \textit{AASS} and \textit{ASAS}. We then compare their performance outcomes to select the superior scheduling policy.

\textbf{Determine $m_a$.}
For illustrative purposes, we focus on optimizing the \textit{ASAS} scheduling strategy. The same methodology can be straightforwardly applied to \textit{AASS}.

Firstly, our optimization focuses on $r_1$ and $m_a$. Given a fixed execution order, we can iteratively compute the key timing variables: $\tau_a^{(t, i)}$, $\tau_s^{(t, i)}$, $\tau_e^{(t, i, j)}$, $\tau_{a2e}^{(t, i, j)}$, and $\tau_{e2a}^{(t, i, j)}$. 

We first examine the timing relationships within the 0‑th layer. Since the derivation for each variable follows a similar pattern, we use \(\tau_{e2a}^{(0, i, j)}\) as an illustrative example. This timestamp depends on the completion of \(i\)  pipeline chunks and additionally on the completion of \(j\) fine‑grained pipeline steps. Its value therefore decomposes into three components: an initial latency, the cumulative delay from the \(r_1\) pipeline, and the cumulative delay from the fine‑grained \(r_2\) pipeline, visualized in Fig.~\ref{fig:split_equ}.


\begin{figure}[t]
    \centering
    \includegraphics[width=0.9\linewidth]{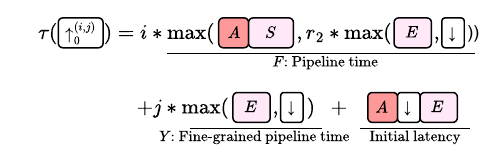}
    \caption{Diagram of the 0-th layer start timestamp $\tau_{e2a}^{(0, i, j)}$, decomposed into three components: pipeline time, fine-grained pipeline time, and initial latency.}
    \label{fig:split_equ}
\end{figure}

Proceeding similarly for all variables, we obtain the complete set of timing expressions for the 0‑th layer:
\begin{equation}
\left\{
\begin{aligned}
\tau_a^{(0, i)} &= i \cdot X(m_a) \\
\tau_s^{(0, i)} &= i \cdot X(m_a) + t_a(m_a) \\
\tau_{a2e}^{(0, i, j)} &= t_a(m_a) + i \cdot F(m_a, m_e) + j \cdot t_{a2e}(m_e) \\
\tau_e^{(0, i, j)} &= t_a(m_a) + t_{a2e}(m_e) + i \cdot F(m_a, m_e) + j \cdot Y(m_e) \\
\tau_{e2a}^{(0, i, j)} &= t_a(m_a) + t_{a2e}(m_e) + t_e(m_e) \\
&\quad + i \cdot F(m_a, m_e) + j \cdot Y(m_e)
\end{aligned},
\right. \notag
\end{equation}
where $X(m_a) = t_a(m_a) + t_s(m_a)$, $Y(m_e) = \max(t_e(m_e), t_{a2e}(m_e))$, and $F(m_a, m_e) = max(X(m_a), r_2 \cdot Y(m_e))$.

For the $t$-th layer, the timing variables $\tau_a^{(t, i)}$, $\tau_s^{(t, i)}$, $\tau_e^{(t, i, j)}$, $\tau_{a2e}^{(t, i, j)}$, and $\tau_{e2a}^{(t, i, j)}$ can be derived based on the corresponding variables from the $(t-1)$-th layer. Specifically, each of them is computed by adding an offset term to their respective $(t-1)$-th counterparts. This offset is given by:
$\max(G(m_a, m_e),\ r_1 \cdot F(m_a, m_e)),$
where  
\begin{equation}
\begin{aligned}
G(m_a, m_e) = & t_a(m_a) + t_{a2e}(m_e) + t_e(m_e) \\
 & + t_{a2e}(m_e) + (r_2 - 1) \cdot Y(m_e).    
\end{aligned}
\end{equation}
Here, $E(m_a, m_e)$ represents the time required to ensure that the GPUs in AG are idle and ready for the next attention segment, while $r_1 \cdot F(m_a, m_e)$ denotes the time required for the output of the expert computation on the 0-th chunk to be sent back. 

We can simplify the optimal objective defined in Eq.~\ref{eq:final-objective} as follows:
\begin{align}
\max_{\substack{r_1, m_a \\ r_2, m_e}} 
& \frac{r_1 \cdot m_a}{
\begin{aligned}
    & (T - 1) \max(G(m_a, m_e), r_1 F(m_a, m_e)) \\
    & + \max(X(m_a), G(m_a, m_e)) + (r_2 - 1) Y(m_e) \\
    & + (r_1 - 1) F(m_a, m_e)
\end{aligned}
}. \label{eq:simplify-objective}
\end{align}





To accelerate the search process, we first identify a crucial property: for a fixed value of $r_1$, the objective function defined in Eq.~\ref{eq:simplify-objective} increases monotonically with respect to $m_a$. To establish this, we employ a two-step proof. First, we demonstrate that for any given pair $(r_1, r_2)$, the objective function in Eq.~\ref{eq:simplify-objective} is monotonically increasing with respect to $m_a$. The detailed proof of this claim is presented below.

\begin{theorem}\label{thm:r_1_and_r_2}
Given pair $(r_1, r_2)$, the objective function in Eq.~\ref{eq:simplify-objective} is monotonically increasing with respect to $ m_a $ .
\end{theorem}

\begin{proof}
To analyze the behavior of the objective function concerning $m_a$, we first establish a direct relationship between $m_e$ and $m_a$. From the constraint $m_a \cdot ag \cdot top_k \cdot S = m_e \cdot r_2 \cdot E$, we can express $m_e$ as a linear function of $m_a$. We have $m_e = k \cdot m_a$, where the constant $k = \frac{ag \cdot top_k \cdot S}{r_2 \cdot E}$.

The component functions $X(m_a)$, $Y(m_e)$, and $E(m_a, m_e)$ are defined as sums and maximums of the base linear performance models $t_a(m_a)$, $t_s(m_a)$, and $t_e(m_e)$. By substituting $m_e = k \cdot m_a$, each of these components becomes a linear or piecewise linear function of $m_a$. Specifically, the denominator of the objective function is constructed from additions and max operations on these functions. Since the sum of linear functions is linear, and the maximum of linear functions is piecewise linear and convex, the entire denominator is a positive, piecewise linear, and convex 
function of $m_a$.


Therefore, the objective function takes the form of $\frac{r_1 m_a}{D(m_a)}$, where $D(m_a)$ is the piecewise linear denominator. 
Within any linear segment of $D(m_a)$, the objective function can be written as $\frac{r_1 m_a}{\alpha_{total} + \beta_{total} m_a}$, where $\alpha_{total}$ and $\beta_{total}$ are positive constants aggregated from the underlying $\alpha$ and $\beta$ parameters of the performance models. To demonstrate its monotonic nature, we can rewrite the previous equation as $\frac{r_1}{\frac{\alpha_{total}}{m_a} + \beta_{total}}$. As $m_a$ increases, the term $\frac{\alpha_{total}}{m_a}$ decreases. This causes the denominator of the overall expression to decrease, which in turn increases the value of the function. Thus, the objective function is monotonically increasing concerning $m_a$ across each linear segment, and therefore, it is monotonically increasing for all $m_a > 0$.
\end{proof}

Next, we extend the result to show that for a fixed value of $r_1$, the objective function in Eq.~\ref{eq:simplify-objective} increases monotonically with respect to $m_a$. This follows directly from Theorem~\ref{thm:r_1_and_r_2}. The detailed proof of this generalized claim is provided below.

\begin{theorem}\label{thm:second_step}
Given $r_1$, the objective function in Eq.~\ref{eq:simplify-objective} is monotonically increasing with respect to $ m_a $.
\end{theorem}

\begin{proof}
Consider any arbitrary value $m_a$ and the corresponding pair $(r_1, r_2^*)$, where
$$
r_2^* = \arg \max_{r_2} \frac{r_1 \cdot m_a}{\max ( \tau_s^{(T, r_1)} + t_s(m_a), \; \tau_{e2a}^{(T, r_1, r_2)} + t_{e2a}(m_e) )},
$$
according to Theorem~\ref{thm:r_1_and_r_2}, for any $m_a' > m_a$, the following inequality holds:
\begin{equation}
    \begin{aligned}
        & \frac{r_1 \cdot m_a'}{\max ( \tau_s^{(T, r_1)} + t_s(m_a'), \; \tau_{e2a}^{(T, r_1, r_2^*)} + t_{e2a}(m_e') )} \\
        > &  \max_{r_2} \frac{r_1 \cdot m_a}{\max ( \tau_s^{(T, r_1)} + t_s(m_a), \; \tau_{e2a}^{(T, r_1, r_2)} + t_{e2a}(m_e) )}. \notag
    \end{aligned}
\end{equation}
It implies that
\begin{equation}
    \begin{aligned}
        & \frac{r_1 \cdot m_a'}{\max ( \tau_s^{(T, r_1)} + t_s(m_a'), \; \tau_{e2a}^{(T, r_1, r_2^*)} + t_{e2a}(m_e') )} \\
        \leq & \max_{r_2} \frac{r_1 \cdot m_a'}{\max ( \tau_s^{(T, r_1)} + t_s(m_a'), \; \tau_{e2a}^{(T, r_1, r_2)} + t_{e2a}(m_e') )}. \notag
    \end{aligned}
\end{equation}
Consequently, the objective function increases monotonically as $m_a$ increases, which completes the proof.
\end{proof}

\begin{figure*}[t]
    \centering
    \includegraphics[width=0.9\linewidth]{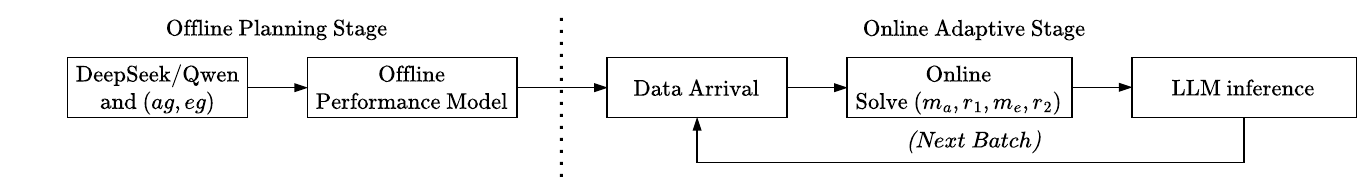}
    \caption{The pipeline of FinDEP, which consists of an offline planning phase and an online adaptive phase.}
    \label{fig:pipeline}
\end{figure*}

\textbf{Determine $r_1$.}
    We now turn our attention to analyzing the behavior of the objective function with respect to the parameter $r_1$. Specifically, we aim to demonstrate that, for a fixed value of $m_a$, the objective function defined in Eq.~\ref{eq:simplify-objective} is monotonically non-decreasing with respect to $r_1$. To establish this result, we adopt a two-step proof strategy analogous to that used in prior analysis.

The first step involves showing that, for any fixed pair $(m_a, r_2)$, the objective function increases or remains constant as $r_1$ increases. The second step mirrors the approach in Theorem~\ref{thm:second_step}. However, due to space constraints, we omit this proof from the paper. In what follows, we focus on formally proving the first step.

\begin{theorem}
Given $(m_a, r_2)$, the objective function in Eq.~\ref{eq:simplify-objective} is monotonically non-decreasing with respect to $r_1$.
\end{theorem}

\begin{proof}
To analyze monotonicity with respect to $r_1$, observe that the objective function takes the form $\frac{r_1 m_a}{D(r_1)}$, where the denominator $D(r_1)$ is piecewise linear in $r_1$. When expressed as $D(r_1) = B r_1 + C$ in each linear segment ($B > 0$), monotonicity depends critically on the sign of the constant term $C$. 
We demonstrate $C \geq 0$, where
\begin{equation}
C = \max(X(m_a), G(m_a, m_e) + (r_2 - 1)Y(m_e)) - F(m_a, m_e).
\end{equation}
First, from the inequality $E(m_a, m_e) \geq Y(m_e)$, we derive:
\begin{equation}
G(m_a, m_e) + (r_2 - 1)Y(m_e) \geq r_2 Y(m_e).
\end{equation}
Consequently, we have:
\begin{equation}
\begin{aligned}
 &\max(X(m_a), G(m_a, m_e) + (r_2 - 1)Y(m_e)) \\
 \geq & \max(X(m_a), r_2 Y(m_e)).
\end{aligned}
\end{equation}
Since $F(m_a, m_e) = \max(X(m_a), r_2 Y(m_e))$, it follows directly that $C \geq 0$.
With $A = m_a > 0$, $B > 0$, and $C \geq 0$, the objective function becomes $\frac{A r_1}{B r_1 + C}$. Its derivative is:
\[
\frac{d}{d r_1} ( \frac{A r_1}{B r_1 + C} ) = \frac{A(B r_1 + C) - A r_1 B}{(B r_1 + C)^2} = \frac{A C}{(B r_1 + C)^2} \geq 0,
\]
since $A > 0$, $C \geq 0$, and the denominator is positive. Therefore, the objective function is monotonically non-decreasing in $r_1$.
\end{proof}

\textbf{Determine $r_2$ and $m_e$.}
The final parameters to verify are \( r_2 \) and \( m_e \). Given \( m_a \) and \( r_1 \), we have only one free variable, as the other is constrained by the relation:  \( m_a\cdot ag \cdot top_k \cdot S / E = m_e \cdot r_2. \)  
To simplify, we express \( m_e(1 / r_2) = (m_a ag \text{top}_k S) / (E \cdot r_2) = k' / r_2 \)  
thereby reducing the problem to solving for \( r_2 \) alone. Fortunately, the objective function is convex with respect to \( 1 / r_2 \).  
\begin{theorem}\label{thm:three_step}  
Given \( r_1 \) and \( m_a \), the objective function in Eq.~\ref{eq:simplify-objective} is convex with respect to \( 1/r_2 \).  
\end{theorem}  


\begin{proof}
To optimize the objective function in Eq.~\ref{eq:simplify-objective}, we express it as the following equivalent form:
\begin{align}
\min_{r_2} & ( (T - 1) \max(G(m_a, m_e), r_1 F(m_a, m_e)) ) \notag \\
& + \max(X(m_a), G(m_a, m_e)) + (r_2 - 1) Y(m_e) \notag \\
& + (r_1 - 1) F(m_a, m_e) \label{equ:solver_r2}.
\end{align}
We aim to prove the convexity of this objective function. Specifically, we need to verify the convexity of the term $(r_2 - 1) Y(m_e)$ and the product $r_2(m_e) Y(m_e)$ within $F(m_a, m_e)$.
The performance models $t_e(1/r_2) = \alpha_e + \beta_e k' / r2$ and $t_{a2e}(1/r_2) = \alpha_{a2e} + \beta_{a2e} k' / r_2 $ are linear functions of $1/r_2$. Since their coefficients are positive, these functions are convex and monotonically increasing. We define $Y(1/r_2) = \max(t_e(1/r_2), t_{a2e}(1/r_2))$. The maximum of linear functions is piecewise linear and convex, and since both are increasing, $Y(1/r_2)$ is non-decreasing, preserving convexity.

Next, the product $r_2 \cdot Y(m_e)$ is the maximum of terms of the form $\alpha r_2 + k'\beta$, which are convex for $r_e > 0$. Hence, $r_2 \cdot Y(m_e)$ is convex.

The function $G(m_a, m_e)$ includes terms like $t_e(1/r_2)$, $t_{a2e}(1/r_2)$, and $(r_2 - 1) Y(1/r_2)$. The linear terms are convex, and $(r_2 - 1) Y(m_e)$, which is a maximum of convex functions, is also convex. Thus, $G(m_a, m_e)$ is convex.

The objective function is the sum of three terms: $(r_1 - 1) F(m_a, m_e)$, which is convex since $r_1 \geq 1$ and $F(m_a, m_e)$ is convex; $(T - 1) \max(G(m_a, m_e), r_1 F(m_a, m_e))$, which is convex since both $G(m_a, m_e)$ and $r_1 F(m_a, m_e)$ are convex, and $T \geq 1$; and $\max(X(m_a), G(m_a, m_e))$, which is convex because it is the maximum of a constant and a convex function. Since the sum of convex functions is convex, the entire objective function is convex with respect to $1/r_2$.
\end{proof}

\begin{algorithm}[t]
\caption{FinDEP Configuration Search}
\label{alg:optimal_config}
\begin{algorithmic}[1]
\Require
    $P, ag, eg, \alpha_*, \beta_*, B, S, H, M, N_{\text{shared}}, E, top_k, T$
\Ensure
    $\text{best\_config} = (m_a, r_1, m_e, r_2, \text{order})$

\State $\text{best\_tps} \gets 0$
\State $\text{best\_config} \gets \emptyset$

\State $r_1' \gets 0$ \Comment{Previous $r_1$}
    
\For{$m_a = \infty$ \textbf{downto} $1$} 
    \State $r_1 \gets \text{getMaxR1}(ag, eg, m_a, P, B, S, H, M, N_{\text{shared}}, E, top_k, T)$ \Comment{Memory-constrained}
    
    \If{$r_1 == 0$ or $r_1 == r_1'$}
        \State \textbf{continue} \Comment{Skip non-Pareto-optimal $(m_a, r_1)$}
    \EndIf
    
    \For{$\text{order} \in \{ASAS, AASS\}$} \Comment{Evaluate both execution orders}
        \State $r_2^*,\ \text{tps} \gets \text{Solve}( \min_{r_2} \text{Eq.~\ref{equ:solver_r2}} )$ \Comment{Returns optimizer and optimal value}
        \State $m_e \gets \frac{m_a \cdot ag \cdot top_k \cdot S}{r_2^* \cdot E}$
        
        \If{$\text{tps} > \text{best\_tps}$}
            \State $\text{best\_tps} \gets \text{tps}$
            \State $\text{best\_config} \gets ( m_a, r_1, m_e, r_2^*, \text{order})$
        \EndIf
    \EndFor
    
    \State $r_1' \gets r_1$
\EndFor

\State \Return $\text{best\_config}$
\end{algorithmic}
\end{algorithm}

\subsection{Algorithm}

Based on the previous analysis, we propose an efficient algorithm to find the near-optimal configuration for $r_1$, $m_a$, $r_2$, and $m_e$, as shown in Algorithm~\ref{alg:optimal_config}. Given a computing order, the algorithm provides the optimal configuration for that order, focusing on maximizing inference throughput.
Specifically, we focus on the Pareto frontier of $(m_a, r_1)$ under memory constraints, respecting the monotonicity of $m_a$ and $r_1$. 
The algorithm iterates over $m_a$ in descending order and calculates the maximum allowable $r_1$ based on memory limits. We skip configurations with the same $r_1$ as the previous iteration to avoid redundancy.

For each unique $(m_a, r_1)$ pair, the algorithm evaluates two execution orders: $\textit{ASAS}$ and $\textit{AASS}$. For each order, we solve a convex optimization problem to find the optimal $r_2$ that maximizes the objective in Eq.~\ref{eq:simplify-objective}. Then, we calculate $m_e$ as:  $m_e = m_a \cdot ag \cdot top_k \cdot S / (r_2 \cdot E)$. The configuration with the highest throughput is returned. This approach efficiently explores the search space and eliminates suboptimal configurations. With the near-optimal solution derived from Algorithm~\ref{alg:optimal_config}, we obtain the fine-grained task schedule in FinDEP for MoE inference.

\textbf{Complexity Analysis}
Our complexity analysis is divided into two steps: first, determining the number of possible $(r_1, m_a)$ positions on the Pareto frontier, and second, analyzing the time spent on convex optimization. Since the memory constraint is $r_1 \cdot m_a \leq M$, where $M$ is the largest micro-batch size that the GPU can hold, the number of distinct values of $m_a = \lfloor \frac{M} {r_1}\rfloor$ corresponds to the number of divisors of $M$, denoted as $d(M)$. The number of divisors grows at most as $O(\sqrt{M})$. Since convex optimization is performed for a single parameter $r_2$, the solver operates quickly. Assuming constant optimization time, denoted as $C$, the overall complexity is $O(C \cdot d(M))$. Given the fast nature of the solver, the inference time is almost unaffected by this process(\S\ref{sec:evaluation}).

\textbf{Online Pipeline} 
Fig.~\ref{fig:pipeline} illustrates our system pipeline, which is bifurcated into offline and online phases. The offline phase handles initialization: we first select the serving model (e.g., DeepSeek or Qwen) and determine the sizes of the Attention Group and Expert Group ($ag, eg$). Subsequently, we utilize an offline performance model to collect the necessary model coefficients and hardware parameters, which serve as inputs for the optimization solver.

The online phase addresses runtime adaptation. As input data shapes are unknown prior to request arrival, configuration decisions must be made in real-time. Upon data arrival, the system executes the lightweight Algorithm~\ref{alg:optimal_config} to rapidly derive the optimal configuration $(m_a, r_1, m_e, r_2, \text{order})$. This approach allows FinDEP to dynamically adapt to varying workloads, achieving superior speedup ratios compared to static settings.
\section{Evaluation}\label{sec:evaluation}

\subsection{Testbeds}
Our experiments leverage four distinct hardware testbeds. Testbed A uses a single node with eight NVIDIA A6000 GPUs, while Testbed B is configured with eight NVIDIA A10 GPUs. Testbed C also employs a single node, equipped with eight NVIDIA H20 GPUs, and Testbed D scales this configuration across four nodes, each containing eight H20 GPUs. Further details regarding the server configuration can be found in Table~\ref{tab:server-config}. Our software environment runs on Ubuntu 22.04, with Python 3.10, CUDA 11.3, PyTorch 2.4, and NCCL 2.27.5.  We implement attention using FlashInfer 0.3.0~\cite{flashinfer}.We implement Attention-to-Expert and Expert-to-Attention transfer atop NCCL.

\begin{table*}[]
	\centering
 \addtolength{\tabcolsep}{-3pt}
 		\caption{The server configurations in our testbeds.}
		\label{tab:server-config}
\begin{tabular}{lllll}
\toprule
            \textbf{Name}    & \textbf{Testbed A}  &  \textbf{Testbed B}    &  \textbf{Testbed C}    &  \textbf{Testbed D}                                                      \\ \midrule 
            Memory & 48GB & 24GB & 96GB & 96GB \\     
\multicolumn{1}{l}{GPU}     & \multicolumn{1}{l}{8x Nvidia RTXA6000 }                                                    & \multicolumn{1}{l}{8x Nvidia A10}   &   \multicolumn{1}{l}{8x Nvidia H20}   &   \multicolumn{1}{l}{32x Nvidia H20}                      \\ 
\multicolumn{1}{l}{Architecture}     & \multicolumn{1}{l}{Ampere }                                                    & \multicolumn{1}{l}{Ampere}   &   \multicolumn{1}{l}{Hopper}   &   \multicolumn{1}{l}{Hopper}                      \\ 
\multicolumn{1}{l}{Boost Clock}     & \multicolumn{1}{l}{1.46GHz}                                                    & \multicolumn{1}{l}{1.41GHz}   &   \multicolumn{1}{l}{1.98GHz}      &   \multicolumn{1}{l}{1.98GHz}                     \\ 
\multicolumn{1}{l}{NVlink}  & \multicolumn{1}{l}{YES}                                                             & \multicolumn{1}{l}{NO}  & \multicolumn{1}{l}{YES}      & \multicolumn{1}{l}{YES}                                                             \\ 
\multicolumn{1}{l}{PCIe}    & \multicolumn{1}{l}{4.0 (x16)}                                                             & \multicolumn{1}{l}{4.0 (x16)}  & \multicolumn{1}{l}{4.0 (x16)}   & \multicolumn{1}{l}{4.0 (x16)}                                                         \\ \bottomrule
\end{tabular}
\end{table*}

\subsection{Verification of Performance Models}~\label{sec:5.2}



\begin{figure}[t]
    \centering

    \begin{subfigure}{0.47\linewidth}
        \centering
        \includegraphics[width=\linewidth]{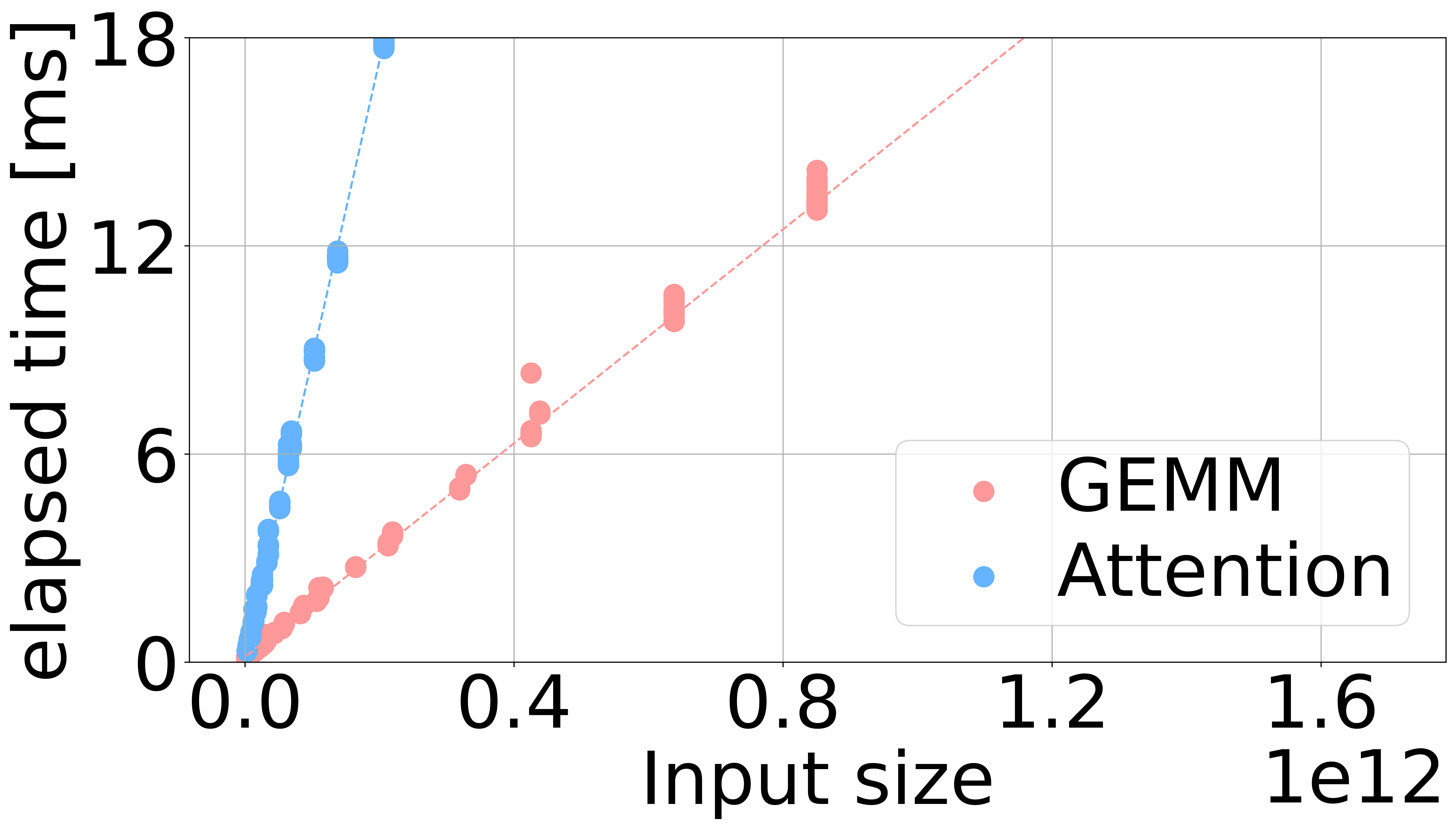}
        \caption{computing on A6000.}
        \label{fig:comp_performance}
    \end{subfigure}
    \begin{subfigure}{0.47\linewidth}
        \centering
        \includegraphics[width=\linewidth]{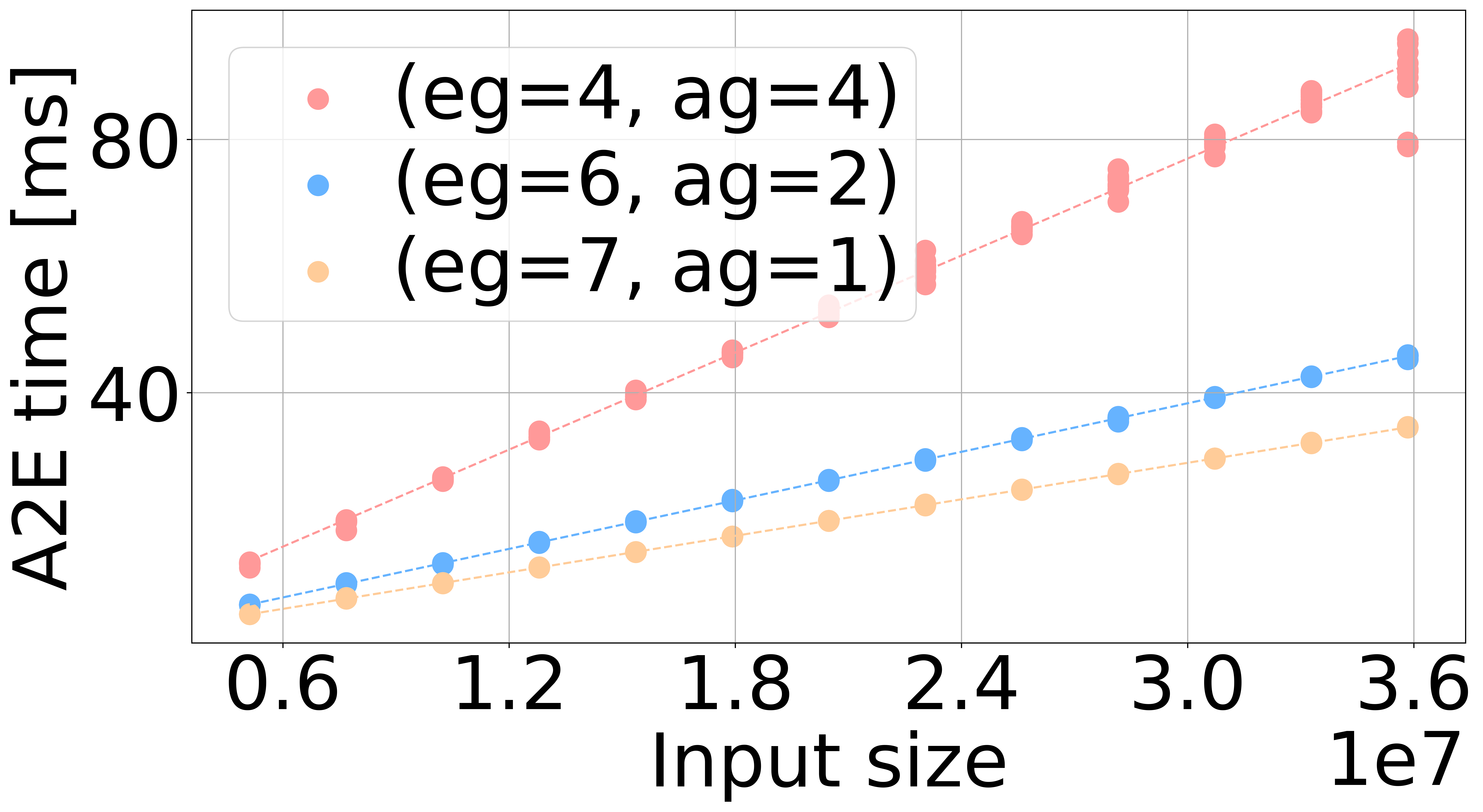}
        \caption{communication on A6000.}
        \label{fig:comm_performance}
    \end{subfigure}
    \caption{Performance models for GEMM, Attention, and communication. Markers represent measured values, while the lines correspond to predicted values with estimated parameters. (a) $\alpha_{gm} = 0.17$ and $\beta_{gm} = 8.59 \times 10^{-11}$. $\alpha_{attn} = 0.15$ and $\beta_{attn} = 1.54 \times 10^{-11}$.  (b) $(\alpha_{a2e}, \beta_{a2e})$ take the following values: $(0.10, 9.61 \times 10^{-7})$ for $(eg = 7, ag = 1)$, $(0.01, 1.28 \times 10^{-6})$ for $(eg = 6, ag = 2)$, and $(0.37, 2.55 \times 10^{-6})$ for $(eg = 4, ag = 4)$.}
    \label{fig:performance}
\end{figure}

We conduct micro-benchmarks to determine the values of $\alpha_{gm}$, $\beta_{gm}$, $\alpha_{attn}$, $\beta_{attn}$, $\alpha_{a2e}$ and $\beta_{a2e}$ before solving the algorithm. For the GEMM component, we test a range of matrix configurations across all matrix sizes encountered in the MLA. This comprehensive testing ensures our model can effectively handle varying configurations. The results, shown in Fig.~\ref{fig:comp_performance}, yield an $R^2$ value of 0.997132.

For the communication component, we separately compute $\alpha_{a2e}$ and $\beta_{a2e}$ for different $ag$ and $eg$ settings, as these parameters are interdependent. Since the time for Expert-to-Attention matches that for Attention-to-Expert, we do not need to rerun the micro-benchmark for the former case. The results, shown in Fig.~\ref{fig:comm_performance}, yield $R^2$ values of 0.999986, 0.999911, and 0.994018, indicating a strong fit. This demonstrates that simple linear models can accurately predict execution time, consistent with findings in prior work on performance modeling~\cite{shi2023pipemoe, gpu_launch, fsmoe}.

We run 30 trials per data point: 10 for warm-up and 20 for statistics. The full micro-benchmark, including Attention, GEMM, and communication steps, takes under 2 minutes.

\subsection{Monotonicity of Throughput with Respect to $m_a$ and $r_1$}

Our analysis reveals a key monotonic relationship: under per-parameter optimization, throughput increases monotonically with respect to \(m_a\) and \(r_1\). Specifically, for a given model and a fixed \((a_g, e_g)\) configuration, if the value of \(m_a\) is held constant, throughput increases as \(r_1\) increases, provided that the \((m_e, r_2)\) pair and computation order are optimized for each specific value of \(r_1\). Conversely, the same monotonic increase holds for \(m_a\) when \(r_1\) is fixed and \((m_e, r_2)\) and the computation order are optimized accordingly.

In this experiment, for each \((m_a, r_1)\) pair, we performed a brute-force search over all \((m_e, r_2)\) values and computation orders to determine the optimal throughput. To accelerate testing, we used a smaller variant of DeepSeek-V2 236B~\cite{deepseek_v2}, keeping all other hyperparameters unchanged and employing only two MoE layers. On Testbed C, we set \((a_g, e_g) = (3, 5)\) and \(S = 2048\text{, } 4096\). On Testbed D, we set \((a_g, e_g) = (8, 24)\) and \(S = 2048\text{, } 4096\).

\textbf{Throughput increases monotonically with $m_a$} 
As shown in Table~\ref{tab:r1}, throughput rises as \(m_a\) increases while \(r_1\) is fixed at 1. This confirms that our theoretical proof aligns with the experimental results. Demonstrating this monotonic relationship allows us to constrain the candidate variable space, thereby speeding up the search process.

\begin{table}[!t]
\centering
\caption{Throughput (tokens/s) of DeepSeek-V2 on Testbed C and Testbed D for varying $m_a$ and sequence length $S$.}
\label{tab:r1}
\begin{tabular}{cccccc}
\toprule
\textbf{Testbed} & $\boldsymbol{S}$ & $\boldsymbol{m_a=1}$ & $\boldsymbol{m_a=2}$ & $\boldsymbol{m_a=4}$ \\
\midrule 
\multirow{2}{*}{Testbed C}  & 2048 & 202.67 & 245.33 & 284.00 \\
\cmidrule(lr){2-5}
& 4096 & 230.12 & 254.84 & 270.35 \\
\cmidrule{1-5}
\multirow{2}{*}{Testbed D} & 2048 & 558.23 & 690.47 & 756.35 \\
\cmidrule(lr){2-5}
& 4096 & 632.41 & 682.49 & 707.57 \\
\bottomrule
\end{tabular}
\end{table}

\textbf{Throughput increases monotonically with $r_1$} 
Similarly, as shown in Table~\ref{tab:ma}, throughput rises as \(r_1\) increases while \(m_a\) is fixed at 1. This result further validates our theoretical proof. Establishing this second monotonic relationship similarly constrains the search space, improving overall optimization efficiency.

\begin{table}[!t]
\centering
\caption{Throughput (tokens/s) of DeepSeek-V2 on Testbed C and Testbed D for varying $r_1$ and sequence length $S$.}
\label{tab:ma}
\begin{tabular}{cccccc}
\toprule
\textbf{Testbed} & $\boldsymbol{S}$ & $\boldsymbol{r_1=1}$ & $\boldsymbol{r_1=2}$ & $\boldsymbol{r_1=4}$ \\
\midrule
 \multirow{2}{*}{Testbed C}  & 2048 & 202.67 & 257.24 & 282.04 \\
\cmidrule(lr){2-5}
& 4096 & 230.12 & 262.62 & 269.92 \\
\cmidrule{1-5}
 \multirow{2}{*}{Testbed D}  & 2048 & 558.23 & 711.36 & 760.48 \\
\cmidrule(lr){2-5}
& 4096 & 632.41 & 714.66 & 735.46 \\
\bottomrule
\end{tabular}
\end{table}





\subsection{Evaluation on Real-World Models}
We evaluate the average end-to-end training iteration time for the small DeepSeek-V2 236B~\cite{deepseek_v3} model, using an 8-layer configuration on testbed A, a 4-layer configuration on testbed B, and a 16-layer configuration on testbed C and D. Additionally, we assess the performance of the small Qwen3-235B-A22B~\cite{qwen3} model, with a 24-layer configuration on Testbed A and a 12-layer configuration on Testbed B and a 48-layer configuration on Testbed C and D. Our approach, FinDEP, is compared against the state-of-the-art PPPipe~\cite{MegaScale_Infer}, for which we provide our own reimplementation to ensure a fair comparison. 

Table~\ref{tab:real_world_combined_modified} reports the average iteration throughput (tokens per second) across different sequence lengths (specifically 1024, 2048, 4096, and 8192), two model backbones (DeepSeek-V2 and Qwen3), and four testbeds (A, B, C, D). Each throughput value represents the average of three independent runs. The data demonstrate that FinDEP consistently outperforms the optimally configured PPPipe across all experimental dimensions. The speedup achieved by FinDEP, indicated in parentheses within the table, ranges from 1.02$\times$ to 1.61$\times$. This performance advantage holds true for varying computational scales (testbeds A through D) and is evident across the full spectrum of tested sequence lengths. When the sequence is very long, FinDEP is much faster (see the bold numbers 1.53$\times$ and 1.61$\times$ in the table). Notably, the solver completes in under 1 second.


\begin{table*}[!t]
\centering
\caption{Average iteration throughput (tokens per second) comparison, where each number is the average of 3 independent runs. The values in brackets represent the speedups achieved by FinDEP compared to PPPipe with optimal $ep$, $dp$, $m_a$, and $r_1$ settings.}
\label{tab:real_world_combined_modified}
\begin{tabular}{ccccccccccc}
\toprule
& & \multicolumn{2}{c}{\textbf{Testbed A}} & \multicolumn{2}{c}{\textbf{Testbed B}} & \multicolumn{2}{c}{\textbf{Testbed C}} & \multicolumn{2}{c}{\textbf{Testbed D}} \\
\cmidrule(lr){3-4} \cmidrule(lr){5-6} \cmidrule(lr){7-8} \cmidrule(lr){9-10}
\textbf{Backbone} & $\boldsymbol{S}$ & \textbf{PPPipe} & \textbf{FinDEP} & \textbf{PPPipe} & \textbf{FinDEP} &  \textbf{PPPipe} & \textbf{FinDEP} &  \textbf{PPPipe} & \textbf{FinDEP} \\
\midrule
\multirow{3}{*}{DeepSeek} & 1024 & 48.50 & 53.40 (1.10$\times$) & 86.70 & 93.04 (1.07$\times$) & 62.31 & 63.35 (1.02$\times$) & 149.58 & 161.50 (1.08$\times$) \\
& 2048 & 46.28 & 50.27 (1.09$\times$) & 81.99 & 86.63 (1.06$\times$) & 56.63 & 58.14 (1.03$\times$) & 134.42 & 150.82 (1.12$\times$) \\
& 4096 & 44.21 & 51.47 (1.16$\times$) & 81.04 & 85.84 (1.06$\times$) & 49.80 & 54.73 (1.10$\times$) & 120.83 & 132.07 (1.10$\times$)  \\
\midrule
\multirow{4}{*}{Qwen} & 1024 & 13.94 & 15.81 (1.13$\times$) & 31.52 & 35.09 (1.11$\times$) & 35.70 & 36.86 ($1.03 \times$)\ & 94.97 & 102.60 (1.08$\times$) \\
& 2048 & 14.00 & 15.85 (1.20$\times$) & 25.46 & 27.39 (1.08$\times$) & 32.78 & 33.50 ($1.02 \times$) & 83.12 & 90.15(1.08$\times$)  \\
& 4096 & 13.80 & 15.55 (1.13$\times$) & 22.48 & 27.64 (1.23$\times$) & 28.01 & 30.06 (1.07$\times$) & 61.59 & 76.53 (1.24$\times$)  \\
& 8192 & 8.57 & 13.14 (\textbf{1.53}$\times$) & 15.98 & 25.71 (\textbf{1.61}$\times$) & 20.14 & 27.12 (\textbf{1.35}$\times$) & 37.19 & 45.26 (\textbf{1.22}$\times$) \\
\bottomrule
\end{tabular}

\end{table*}



\begin{table*}[!t]
\centering
\caption{Average iteration throughput (tokens per second) comparison. The values in brackets represent the speedups achieved by our FinDEP compared to PPPipe with given $eg$, $ag$ settings.}
\label{tab:exp_online_combined}
\begin{tabular}{ccccccccccc}
\toprule
 &  & \multicolumn{2}{c}{\textbf{Testbed A}} & \multicolumn{2}{c}{\textbf{Testbed B}} & \multicolumn{2}{c}{\textbf{Testbed C}} & \multicolumn{2}{c}{\textbf{Testbed D}} \\
\cmidrule(lr){3-4} \cmidrule(lr){5-6} \cmidrule(lr){7-8} \cmidrule(lr){9-10}
\textbf{Backbone} & \textbf{Tokens} & \textbf{PPPipe} & \textbf{FinDEP} & \textbf{PPPipe} & \textbf{FinDEP} & \textbf{PPPipe} & \textbf{FinDEP} & \textbf{PPPipe} & \textbf{FinDEP} \\
\midrule
\multirow{2}{*}{DeepSeek} & 3072 & 28.23 & 29.34 (1.04$\times$) & 44.66 & 50.13 (1.12$\times$)  & 30.24 &  31.13 (1.03$\times$) & 98.80 & 121.07 (1.23$\times$) \\
& 6144 & 41.88 & 44.25 (1.06$\times$) & 69.64 & 75.99 (1.09$\times$)  & 36.67 & 38.13 (1.04$\times$) & 124.69 & 142.36 (1.14$\times$)  \\
\midrule
\multirow{2}{*}{Qwen} & 3072 & 9.14 & 10.95 (1.20$\times$) & 16.24 & 18.56 (1.14$\times$) & 19.15 & 19.16 (1.00$\times$) & 40.94 & 50.71 (1.24$\times$) \\
& 6144 & 13.54 & 15.28 (1.13$\times$) & 22.71 & 30.43 (1.09$\times$) & 30.19 & 30.43 (1.01$\times$) & 67.07 & 78.69 (1.17$\times$) \\
\bottomrule
\end{tabular}
\end{table*}


\textbf{Discussion.} In our configuration of testbed A with the DeepSeek backbone, we observe that FinDEP effectively hides communication costs, approaching near-optimal performance. Compared to PPPipe under the same conditions (e.g., $(eg, ag)$), FinDEP reduces communication by 1.7$\times$ as shown in Table~\ref{tab:deepseek_comparison}. This indicates that, for shorter sequences, communication optimizations offer limited improvement. However, for longer sequences, communication becomes the primary bottleneck. For instance, with a sequence length of 4096, there is a 25.87 ms gap where computation and communication do not overlap. This emphasizes the near-optimal performance of our solution.

\begin{table}[!t]
\centering
\caption{Non-overlapped communication time for naive DEP (Naive-DEP) without pipelining, PPPipe, and FinDEP in DeepSeek-V2 on testbed A.}
\label{tab:deepseek_comparison}
\begin{tabular}{cccc}
\toprule
$\boldsymbol{S}$ & \textbf{Naive-DEP} & \textbf{PPPipe} & \textbf{FinDEP} \\
\midrule
4096 & 905.49ms & 528.94ms & 309.81ms \\
2048 & 536.22ms & 144.32ms & 52.60ms \\
1024 & 194.95ms & 188.65ms & 97.33ms \\
\bottomrule
\end{tabular}
\end{table}

\subsection{Evaluation on Online Settings}

In the online setting, reboot costs limit frequent changes to $ag$ and $eg$. Additionally, the unpredictable user prompt length (i.e., sequence length) complicates DEP deployment. However, our fast solver addresses this by quickly adjusting $r_1$, $r_2$, and the execution order after receiving the prompt length. We evaluate FinDEP with the following configurations: for DeepSeek-V2, $(ag, eg) = (3, 5)$, and for Qwen3-MoE, $(ag, eg) = (4, 4)$ on Testbeds A, B, and C.
For Testbed D, we set $(ag,eg) = (8, 24)$ for both DeepSeek-V2 and Qwen3-MoE. Two scenarios highlight differences in the mean number of arriving tokens. Table~\ref{tab:exp_online_combined} shows that our FinDEP, using the fast solver in Algorithm~\ref{alg:optimal_config}, outperforms the static schedule with the best PPPipe configuration at a sequence length of 2048. By adjusting $r_2$ and $r_1$, we improve the throughput up to 1.20$\times$.


\textbf{Discussion.} In the configuration utilizing the Qwen backbone on Testbed C, we observe that FinDEP does not achieve significant performance gains over PPPipe. As indicated in Table~\ref{tab:exp_online_combined}, FinDEP attains only 1.0$\times$ to 1.1$\times$ the throughput of PPPipe under identical $(ag, eg)$ settings. This result aligns with the expectations of Amdahl's Law, which bounds the maximum speedup achievable by optimizing only a portion of the system. Specifically, the high-bandwidth NVLink interconnect on the H20 GPUs shown in Table~\ref{tab:server-config} renders communication time a comparatively minor component of total runtime. Consequently, further optimization of the execution schedule yields diminishing returns, as the system is primarily constrained by other computational factors.

In contrast, performance improves substantially on Testbed D, where communication and computation overheads are more balanced. Communication overhead increases relative to Testbed C, while per-GPU computation time decreases because experts are distributed across more GPUs. With this improved balance, FinDEP’s throughput increases by up to 1.24$\times$ compared to PPPipe. However, at an extremely large scale, communication would again dominate end-to-end execution time. In that scenario, the relative improvement from schedule optimization would diminish because the proportion of time spent on non-accelerated components would increase once more.


\section{Related Work}

Distributed deep learning systems enhance the inference performance of MoE Large Language Models (LLMs) primarily through a triple strategy: one approach involves offloading and optimizing computation, where large model components or tasks are moved to the CPU, as demonstrated in works like~\cite{mixtral_offloading, expertflow}, or through dedicated computational optimizations such as those found in~\cite{moe_lightning, scheinfer, fiddler}, effectively addressing critical GPU memory constraints and boosting overall throughput. A second, parallel strategy focuses on minimizing expert decoding latency by identifying and duplicating frequently utilized "hot" experts across different resources, a technique leveraged by systems like~\cite{deepseek_v3, fastermoe, flexmoe} to ensure quicker access and processing for high-demand experts. Finally, the third key method employs model quantization~\cite{stepfun, deepseek_v3, PowerInfer}, which significantly reduces the data precision of the model weights often down to 4 bit or 8 bit, thereby shrinking the required communication volume between devices at the cost of a minor, acceptable trade-off in model accuracy or performance, ultimately yielding substantial gains in network efficiency.


Disaggregation is commonly used in LLM serving architectures to optimize inference performance in key ways~\cite{TetriInfer, distserve, kvcache, vllm, sglang}. For example, DistServe~\cite{distserve} disaggregates prefill and decode computations onto separate GPUs, boosting parallelism and improving resource allocation for better performance.  Building on this, recent works have pushed for physical disaggregation. Mooncake~\cite{mooncake} utilizes a disaggregated architecture that separates the KVCache pool from the inference engines, leveraging high-speed interconnects to enable stateless inference workers.



\section{Conclusion}

In this paper, we propose FinDEP, a fine-grained task scheduling framework designed to optimize MoE inference under disaggregated expert parallelism. By partitioning computation and communication into smaller tasks and formulating a formal optimization problem, FinDEP maximizes task overlap and resource utilization. We evaluate FinDEP across four GPU testbeds, including a large-scale 32-GPU system, using representative MoE backbones such as DeepSeek-V2 and Qwen3-MoE. Experimental results demonstrate that FinDEP achieves significant performance gains, providing speedups of up to $1.61\times$ over the best-configured PPPipe algorithm. Notably, on the 32-GPU system, FinDEP still delivers a robust speedup of up to $1.24\times$ in offline scenarios. Furthermore, our solver derives near-optimal configurations in under one second, enabling FinDEP to adapt in real-time to dynamic workloads.


\bibliographystyle{ACM-Reference-Format}
\bibliography{sample-base}

\end{document}